\documentclass[12pt]{article}
\usepackage[utf8]{inputenc}
\usepackage{braket}
\usepackage{amsmath}
\usepackage{amssymb}
\usepackage{wrapfig}
\usepackage{graphicx}
\usepackage{multirow}
\usepackage{amsthm}
\usepackage{subcaption}
\usepackage[left=3cm,right=3cm,top=2cm,bottom=2cm]{geometry}
\usepackage{appendix}
\usepackage{hyperref}

\allowdisplaybreaks
\usepackage{xcolor}

\newcommand{\be}{\begin{equation}}
\newcommand{\ee}{\end{equation}}
\newcommand{\LL}{\mathfrak{L}_1}
\newcommand{\GG}{\mathfrak{M}}

\newtheorem{theorem}{Theorem}
\newtheorem{remark}{Remark}
\newtheorem{corollary}{Corollary}
\newtheorem{Lemma}{Lemma}

\title{Control landscape of measurement-assisted transition probability for a three-level quantum system with dynamical symmetry}
\author{Maria Elovenkova and Alexander Pechen}
\date{}

\begin{document}

\maketitle

Department of Mathematical Methods for Quantum Technologies, Steklov Mathematical Institute of Russian Academy of Sciences, 8 Gubkina str., Moscow, Russia

\begin{abstract}
Quantum systems with dynamical symmetries have conserved quantities which are preserved under coherent controls. Therefore such systems can not be completely controlled by means of only coherent control. In particular, for such systems maximal transition probability between some pair of states over all coherent controls can be less than one. However, incoherent control can break this dynamical symmetry and increase the maximal attainable transition probability. Simplest example of such situation occurs in a three-level quantum system with dynamical symmetry, for which maximal probability of transition between the ground and the intermediate state by only coherent control is $1/2$, and by coherent control assisted by incoherent control implemented by non-selective measurement of the ground state is about $0.687$, as was previously analytically computed. In this work we study and completely characterize all critical points of the kinematic quantum control landscape for this measurement-assisted transition probability, which is considered as a function of the kinematic control parameters (Euler angles). This used in this work measurement-driven control is different both from quantum feedback and Zeno-type control. We show that all critical points are global maxima, global minima, saddle points and second order traps. For comparison, we study the transition probability between the ground and highest excited state, as well as the case when both these transition probabilities are assisted by incoherent control implemented by measurement of the intermediate state. 
\end{abstract}


\section{Introduction}

Quantum control attracts high interest due to fundamental reasons related with ability to manipulate quantum systems and due to various existing and prospective applications in quantum technologies~\cite{KochEPJ2022}. Coherent control is realized by lasers or electromagnetic field~\cite{RiceBook,TannorBook,Shapiro_Brumer_book}. Incoherent control can be realized by various forms of the environment. One particular form of incoherent control is control by {\it back-action of non-selective quantum measurements}, either projective or continuous. The general formulation of assisting coherent control of quantum systems by back-action of non-selective quantum measurements was developed in~\cite{PechenRabitz2006}. Various incoherent control strategies using back-action of quantum measurements were applied to many problems such as enhancing controllability of systems with dynamical symmetry~\cite{Shuang2008}, developing incoherent control schemes for quantum systems with wavefunction-controllable subspaces~\cite{DongIEEE2008}, using measurements to enhance controllability~\cite{Sugny_Kontz_2008}, manipulating a qubit~\cite{Blok2014}, incoherently controlling the retinal isomerization in rhodopsin~\cite{LucasPRL2014}, controlling transitions in the Landau--Zener system~\cite{PechenTrushechkin2015}, modifying decay rates of excited states in open quantum systems~\cite{Zhang_Ai_Li_Xu_Sun_2014}, 
inducing Boolean dynamics for open quantum networks~\cite{QiIEEE2023}, steering quantum states by
sequences of weak blind measurements~\cite{KumarPRA2022}, incoherently controlling a qubit using the quantum Zeno effect~\cite{HacohenPRL2018}, decoherence-assisted optimal quantum state preparation~\cite{CejnarPRA2023}, building quantum feedback-like models in photosynthesis~\cite{KozyrevPRA2022}, controlled quantum state transfer based on the optimal measurement~\cite{Harraz2017}, etc. Induced by energy measurement quantum damping of position for an atom bouncing on a reflecting surface in the presence of a homogeneous gravitational field was investigated~\cite{ViolaPRA1996}. Measurement induced decoherence of a trapped ions prepared in nonclassical motional state was shown to inhibit the internal population dynamics and damp the vibrational motion~\cite{ViolaPRA1997}. Mapping an unknown mixed quantum state onto a known pure state using sequential measurements of two noncommuting observables only and without the use of unitary transformations was studied~\cite{RoaPRA2006}. Zeno type effects in a superconducting qubit in the presence of structured noise baths and variable measurement rates were used to suppress and accelerate the qubit decay~\cite{HarringtonPRL2017}. Quantum measurements were used for controlling weak values in the Cheshire Cat effect in open quantum systems~\cite{DajkaQR2019}. A quantum clock driven by entropy reduction through measurement was proposed~\cite{Gangat2021}. Zeno effect was used for demonstrating a universal control between non-interacting qubits~\cite{BlumenthalNPJ2022} and generating multi-qubit entangling Zeno gate~\cite{Lewalle2022}. Measurement-based estimator scheme for continuous quantum error correction was proposed~\cite{BorahPRR2022}. Steering a two-level system using weak measurements in addition to the system Hamiltonian was shown to allow for targeting any pure or mixed state~\cite{KumarPRA2022}. Stroboscopic driving combined with repeated measurement-like interactions with an external spectator system was applied for optimal quantum state preparation~\cite{CejnarPRA2023}. Measurement-based feedback approach with reinforcement learning exploiting the non-linearity of weak measurements with a coherent driving was proposed to prepare cavity Fock state superpositions~\cite{Perret2023}. Measurement-based control was applied for finding minimum energy eigenstate of a given energy function~\cite{Clausen2023}. Control of two-level systems by using unital positive maps associated with quantum Lotka–Volterra operators was considered ~\cite{MukhamedovSymmetry2022}. Projective quantum measurement induces the collapse of the wavefunction onto eigenstate of the observable which corresponds to the observed eigenvalue. Another paradigm is protective measurements which are able to preserve the coherence of the quantum state during the measurement process and allows to extract the expectation value of an observable by measuring a single quantum system~\cite{Piacentini2017}.

A typical quantum control problem is formulated as maximization of some objective functional (fidelity) which represents the desired property of the system which should be optimized. An important topic is the analysis of the dynamic and kinematic quantum control landscapes~\cite{Rabitz2004,Pechen2011,deFouquieresSchirmer}. The dynamical control landscape is defined as graph of the objective functional as a functional of the control. The kinematic control landscape is formed by representing the dynamics of the controlled quantum system by transformation maps (e.g., by unitary transformations or Kraus maps) instead of differential evolution equations and considering parameters of these maps as controls. 

Key points of the control landscape are global maxima (for maximization of the objective) or global minima (for minimization of the objective). {\it Traps} are local but not global maxima (resp., minima) of the control landscape. An important type of critical points are {\it second-order traps} --- critical points where Hessian of the objective is negative (resp., positive) semi-definite, while not strictly negative (resp., positive). At such points, the objective still may grow (resp., decrease) along directions in the control space which correspond to zero eigenvalues of the Hessian, but the growth is slower than second order of small variation of the control. Second order traps were introduced and studied in~\cite{Pechen2011} (see also~\cite{VolkovUMN2023}). The importance of the analysis of kinematic control landscapes was shown for closed systems in~\cite{Rabitz2004}. For open quantum systems the analysis of the kinematic control landscape was performed in~\cite{PechenBrifPRA2010} and unified analysis for both quantum and classical systems in~\cite{PechenEPL2010}. Applications to problems in chemistry were found~\cite{C0SC00425A,C1CP20353C}.

\begin{figure}[t]
\centering
\includegraphics[width=.2\textwidth]{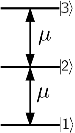}
\caption{Energy level structure and allowed transitions between the energy levels of the three-level quantum system with dynamicsl symmetry. Maximum of the transition probability $|1\rangle\to|2\rangle$ by only coherent control is $1/2$.}
\label{Fig1}
\end{figure}

Quantum systems with dynamical symmetries are of particular interest in quantum control~\cite{Hioe_1983,Turinici_Rabitz_2001,ShuangPRA2008,Sugny_Kontz_2008,Engelhardt_Cao_2020}. Such systems are not completely controllable by means of coherent control, while they may have no Hermitian observables invariant under {\it all} controls (while for any particular control such invariant observable may exist). Instead of that they may have some different conservation laws. This symmetry can be broken using some kind of incoherent control, as was done for example using back-action of non-selective quantum measurements in~\cite{Shuang2008} for a three-level system with dynamical symmetry (see Fig.~\ref{Fig1}). For this system, maximum transition probability $P_{1\to 2}$ between the ground and the intermediate state using only coherent control was shown to be $1/2$~\cite{Turinici_Rabitz_2001}. In~\cite{ShuangPRA2008} maximum of the transition probability when coherent control is assisted by non-selective measurement of population of the ground state was computed and shown to be $0.004(\sqrt{393-48\sqrt{6}}+138+7\sqrt{6})\approx 0.687$. An important analysis of this system was performed in~\cite{Sugny_Kontz_2008}, where the following problem was considered and solved: find an observable such that its von Neumann measurement along with coherent control allows to completely steer the initial state into a target state. In this case the choice of the measured observable is considered as control according to~\cite{PechenRabitz2006}.

In this work, we consider this three-level quantum system in more details. We find and characterize all critical points of the transition probability $P_{1\to 2}$ and in this sense we completely describe its kinematic quantum control landscape, which we show consists of global maxima, global minima, saddle points and second order traps. For comparison, we  compute maximum of the measurement-assisted transition probability $P_{1\to 3}$ between the ground and the upper excited state an well as maximum of both these transition probabilities assisted by measurement of the intermediate state $|2\rangle$.

The structure of this work is the following. Sec.~\ref{Sec:MeasurementControl} reminds general idea of control by back-action of quantum measurements. In Sec.~\ref{Sec2} we provide the formulation of the considered problem. Sec.~\ref{Sec3} contains necessary results from spin-1 representation theory. Sec.~\ref{Sec4} contains preliminary computation of the measurement-assisted transition probabilities in the kinematic representation of controls. Sec.~\ref{Sec5} contains a detailed analysis of the kinematic control landscape for transition probability $P_{1\to 2}$. Sec.~\ref{Sec6} contains computation of maximal values of the transition probability $P_{1\to 3}$ assisted by measurement of the state $|1\rangle$ or $|2\rangle$ and transition probability $P_{1\to 2}$ assisted by measurement of the state $|2\rangle$. Discussion Section~\ref{SecConclusions} summarizes the results.

\section{Measurement-assisted quantum control}\label{Sec:MeasurementControl}

In this section we discuss in details the scheme of quantum control by back-action of non-selective quantum measurements in the formulaintroduced in~\cite{PechenRabitz2006}.

Measurements performed on a quantum system generally are probabilistic --- measurement outcomes for identical measurements performed  on identically prepared quantum systems vary and are obtained with some probabilities. Another important difference of quantum measurement from measurements performed on classical systems is that measurements on quantum systems not only allow to extract some information about the system but often change the state of the system, and such state change can be significant. 

Consider non-selective measurement of a system observable $O=\sum_i \lambda_i P_i$ (a Hermitian operator) with eigenvalues $\lambda_i\in\mathbb R$ and spectral projectors $P_i$, $P_i=P_i^\dagger$, $P_i^2=P_i$. Eigenvalues $\lambda_i$ are the possible measurement outcomes. If $\rho$ is density matrix of the system before the measurement, then measurement outcome $\lambda_i$ is obtained with probability $p_i={\rm Tr}\,\rho P_i$. If the observer looks at the measurement apparatus and reads the measurement outcome $\lambda_i$, then the system state immediately after the measurement will be transformed to $P_i\rho P_i/p_i$. This reduction is known as collapse of the wave function. Non-selective measurement corresponds to situation when the <<observer>> does not look at the measurement apparatus and does not read the measurement outcome. Outcome for this situation can be described as average of the realization of all possible measurement outcomes with their corresponding probabilities. In this situation the measurement outcome is the average value $\langle O\rangle=\sum_i p_i\lambda_i={\rm Tr}\,\rho O $ and according to the von Neumann-L\"uders postulate~\cite{Luders1951,vonNeumann1955} non-selective measurement of the observable $O$ transforms the density matrix as 
\begin{equation}\label{Eq:Measurement1}
\rho\to {\cal M}_{O}(\rho)= \sum\limits_i P_i \rho P_i.
\end{equation}
In particular, the non-selective measurement of population $O=P_\psi=|\psi\rangle\langle\psi|$ of some state $|\psi\rangle$  transforms the density matrix as
\begin{equation}\label{Eq:Measurement2}
\rho\to {\cal M}_{P_\psi}(\rho)= P_\psi \rho P_\psi+(\mathbb{I}-P_\psi)\rho (\mathbb{I}-P_\psi),
\end{equation}
where $\mathbb{I}$ is the identity operator. In quantum control, quantum measurements can be used for real-time feedback which was developed for models of quantum optics in~\cite{WisemanPRL1993,WisemanPRA1994,WisemanMilburnBook}. In this scheme, some control (e.g., a shaped laser field) continuously acts on the controlled quantum system, during this process a discrete or continuous selective measurement of certain observable is performed, and the measurement outcome is processed to modify the applied control in the real time. Various achievements were made in quantum feedback control. Coherent quantum feedback strategy was proposed~\cite{LloydPRA2000}. The problem of quantum feedback control in the context of established formulations of classical control theory was discussed~\cite{DohertyPRA2000}. Hamilton--Jacobi--Bellman equation for quantum optimal feedback control was derived~\cite{Gough2005}. The ability to simulate universal dynamical control through the repeated application of only two coherent control operations and a simple “Yes-No” measurement was shown~\cite{LloydPRA2001}. For controlling position and momentum observables, experimental measurement-based quantum control of the motion of a millimetre-sized membrane resonator was demonstrated~\cite{RossiNature2018}. An approach to information transfer in spintronics networks via the design of time-invariant feedback control laws without recourse to dynamic control was proposed~\cite{SchirmerIEEE2018}. Applications of quantum feedback to quantum engineering are discussed in the review~\cite{Gough2012}. A measurement-based feedback cooling operating deep in the sideband-unresolved limit was demonstrated~\cite{Guo2023}. Feedback quantum control with deep reinforcement learning was applied to driving a system with double well potential with high fidelity toward the ground state~\cite{BorahPRL2021}. Feedback control was proposed for controlling a solid-state qubit~\cite{RuskovPRB2002}, non-linear systems~\cite{Jacobs2007}, quantum state manipulation~\cite{JamesPRA2014}, multi-qubit entanglement generation~\cite{ZhangPRA2020}, making coupled-qubit-based thermal machines~\cite{BhandariEntropy2023}, steering to the minimum energy eigenstate of an energy function in variational quantum algorithms~\cite{Clausen2023}, etc.

Real-time feedback control can potentially be very powerful, but difficult to realize experimentally due to necessity of fast real-time processing of measurement outcome. However, since non-selective quantum measurements themselves modify the state of the system, they potentially can be used for control even without feedback. Such scheme has an advantage of not using fast real-time processing of measurement outcome, but at the cost of necessity to measure various quantum observables. The mathematical formulation of using non-selective quantum measurements with or without coherent quantum control was developed in~\cite{PechenRabitz2006}, where apart of the general formulation the two-level case was completely analytically solved --- for this case optimal measured observables and maximum attained fidelity for maximizing transition probability in a two-level quantum system by any given number of non-selective measurements were analytically computed. The general formulation includes coherent or incoherent control during time intervals $[0,t_1),\dots, (t_k,t_{k+1}),\dots, (t_{K-1},t_K=T]$ inducing on each interval the CPTP dynamical map $\Phi_{(t_k,t_{k+1})}$, and non-selective measurements of observables $O_k$ at time moments $t_k$. The overall evolution of the initial density matrix $\rho_0$ will be
\begin{equation}
	\rho(T)=\Phi_{(t_N,T]}\circ{\cal M}_{Q_N}\circ
	\Phi_{(t_{N-1},t_N)}\circ{\cal M}_{Q_{N-1}}\circ \dots\circ
	\Phi_{(t_1,t_2)}\circ{\cal M}_{Q_1}\circ
	\Phi_{[0,t_1)}(\rho_0).\label{eq2}
\end{equation}
The control goal is to optimize, in addition to coherent and incoherent controls, also the measured observables $O_k$ to maximize a given control objective (e.g., to realize an optimal state transfer).

Another view on this type of quantum control can be described via famous Zeno and anti-Zeno effects. Quantum Zeno effect, predicted by L. Khalfin~\cite{Khalfin1,Khalfin2} and B. Misra and E.~C.~G.~Sudarshan~\cite{MisraSudarshaJMP1977}, states that making a continuous non-selective measurement of population of some state of a quantum system, which has its own internal evolution, leads to freezing the system in the measured state, independent of the internal system evolution.  It reminds the arrow paradox (aporia) of Greek philosopher Zeno of Elea, where for observation of a flying arrow at any one (duration-less) instant of time, the arrow is neither moving to where it is, nor to where it is not. Then it is concluded that since at every instant of time there is no motion occurring, if the arrow is motionless at every instant and time would be entirely composed of instants, then motion of the arrow would be impossible. In quantum anti-Zeno (or dynamical Zeno) effect~\cite{BalachandranRoyPRL2000} one performs a non-selective measurement of some time-dependent state $P_t=|\psi(t)\rangle\langle\psi(t)|$ on the quantum system which evolves with its internal dynamics. As a result, in the limit of continuous measurements independently of the internal dynamics the system will follow the time-dependent state of the measured observable, so that the system density matrix at time $t$ will be $|\psi(t)\rangle\langle\psi(t)|$. The anti-Zeno effect can be used for controlling quantum system if one considers the time-dependent measured observable as control. If $|\psi(T)\rangle =|\psi_{\rm target}\rangle$, then non-selective measurement of $P_t$ via the anti-Zeno dynamics transfers the system to the target state $|\psi_{\rm target}\rangle$, thereby realizing state transfer of the system. Complete realization of such state transfer can be done in the limit of continuous measurement of the observable $P_t$, that can be difficult to perform. Therefore a natural question is how to produce best approximation of the anti-Zeno dynamics by a fixed number $N$ of quantum measurements. For the two-level case, such optimal approximation of anti-Zeno effect by $N$ quantum measurements was found~\cite{PechenRabitz2006}. Maximum transition probability was computed to be
\[
P^{\max}_{i\to f}(N)=\frac{1}{2}
\Bigl[1+\Bigl(\cos\frac{\Delta\varphi}{N+1}\Bigr)^{N+1}\Bigr]<1,
\]
where $\Delta\varphi$  is the angle between Bloch vectors of the initial and target states. The anti-Zeno effect is obtained in the limit of an infinite number of measurements when the interval between any two consecutive measurements tends to zero, $\lim\limits_{N\to\infty}P^{\max}_{i\to f}(N)=1$. 

Thus the considered here based on~\cite{PechenRabitz2006} scheme of control by using back-action of non-selective measurements is different both from quantum feedback control and from Zeno-type control. From feedback it is different since it does not read the measurement outcome. From Zeno-type control it is different since it uses a finite number of measurements. These differences lead to advantages and disadvantages. Advantages are in the simpler experimental realization, since there is no need of real-time feedback and no need of continuous measurements. The potential disadvantage is in the less efficiency of control. How significant is this decrease of the efficiency of course depends on the particular control problem and sufficient level of fidelity which one needs to obtain.

\section{Formulation of the problem}\label{Sec2}

We consider a three-level quantum system with dynamical symmetry, with the free and interaction Hamiltonians defined as (we set Planck constant $\hbar=1$)
\begin{equation}\label{eq:Ham}
H_0=\begin{pmatrix}
0 & 0 & 0\\
0 & 1 & 0\\
0 & 0 & 2
\end{pmatrix},\quad
V=
\mu\begin{pmatrix}
0 & 1 & 0\\
1 & 0 & 1\\
0 & 1 & 0
\end{pmatrix}.
\end{equation}
The energy level structure of this system is shown on Fig.~\ref{Fig1}. Without loss of generality one can set $\mu=1$. Physically, such Hamiltonian represents a three-level system whose interaction with the laser field consists of a dipolar interaction with constant dipolar terms coupling only neighboring states, and is equivalent to a spin-1 particle which was studied in~\cite{Hioe_1983}. Coherent dynamics of more general spin-like quantum systems was studied in~\cite{CookPRA1979}. Controllability of such system was studied in~\cite{Turinici_Rabitz_2001}, where it was shown that such system is not completely controllable by only coherent control. As mentioned above, maximum transition probability for coherent control assisted by single von Neumann measurement of the population of the ground state for this system is  $0.004(\sqrt{393-48\sqrt{6}}+138+7\sqrt{6})\approx 0.687$~\cite{ShuangPRA2008}. In~\cite{Sugny_Kontz_2008}, the problem of funding an observable such that its von Neumann measurement along with coherent control allows to completely steer the initial state into a target state was solved.

Unitary evolution operator of this system evolving under the action of a coherent control $f\in L^2([0,T];\mathbb R)$ satisfies the Schr\"odinger equation
\begin{equation} \label{SchrEq}
    \frac{dU_t^f}{dt}=-i(H_0+f(t)V)U_t^f,\quad U_{t=0}^f= \mathbb{I}\,.
\end{equation}
where $\mathbb I$ is $3\times 3$ identity matrix. Time-evolved pure state of this system can be written as
\[
|\psi(t)\rangle =U^{f}_t|\psi(0)\rangle=C_1(t)|1\rangle+C_2(t)|2\rangle+C_3(t)|3\rangle,\quad C_i(t)\in\mathbb C,\quad i=1,2,3
\]
The coefficients in this expansion satisfy the normalization condition $|C_1(t)|^2+|C_2(t)|^2+|C_3(t)|^2=1$. 

While this system is not controllable, it has no Hermitian observable which is conserved by {\it all} controls, i.e. there is no Hermitian observable $O$ such that for any $f$ $[H_0+f(t)V,O]=0$ (if such an observable would existed, its average value would be determined by the initial state and could not be modified by any control). Indeed, this condition can be satisfied for any $f$ only if $[H_0,O]=[V,O]=0$, that holds only for $O={\rm const}\cdot \mathbb I$ as can be directly checked. For any particular control $f$, an invariant operator exists, as was shown in~\cite{Lai1996} where such operator was even explicitly constructed; however, this invariant operator depends on $f$. Instead of absence of $f$-independent conserved Hermitian operator, it is known~\cite{Turinici_Rabitz_2001} that this system has the conservation law 
\be\label{Eq:Constraint}
\left|C_1(t)C_2(t)-\frac{C_3(t)^2}{2}\right|={\rm Const}\,.
\ee
Note that the quantity under modulus in this constraint is also conserved. This conservation law, as was shown in~\cite{Turinici_Rabitz_2001}, limits maximum of the transition probability $P_{1\to 2}(f)$ from state $|1\rangle$ to state $|2\rangle$ over all coherent controls $f$ to be
\[
\max\limits_f P_{1\to 2}(f)=\frac12.
\]
The transition probability is computed as $P_{1\to 2}(f)=|C_2(T)|^2$ for a large enough $T$, where the coefficients satisfy the initial condition $C_1(0)=1$, $C_2(0)=C_3(0)=0$. 

Consider now non-selective measurement of a system observable $O$ (a Hermitian operator) with spectral projectors $P_i$. If $\rho$ is density matrix of the system before the measurement, then non-selective measurement of $O$ transforms the density according to Eq.~(\ref{Eq:Measurement1}).
In particular, non-selective measurement of population $O=P_\psi=|\psi\rangle\langle\psi|$ of some state $|\psi\rangle$  transforms the density matrix according to Eq.~(\ref{Eq:Measurement2}).

Suppose that the quantum system~(\ref{eq:Ham}) evolves under the action of a coherent control $f_1$ during time interval $[-T_1,0)$ with some sufficiently long time $T_1\ge T^*$, then at $t=0$ non-selective measurement of population of some state $|\psi\rangle$ is performed, and finally during time interval $(0,T_2]$ with $T_2\ge T^*$ the system again evolves under the action of some coherent control $f_2$. Here $T^*$ is the minimal time during which all allowed unitary evolutions can be generated. Then its final density matrix at $t=T_2$ will be 
\begin{equation}\label{state}
\rho(T_2)=U_{f_2}{\cal M}_{P_{\psi}}[U_{f_1}\rho(-T_1) U^\dagger_{f_1}] U^\dagger_{f_2}.
\end{equation}

In particular, probabilities of transitions from state $|1\rangle$ to state $|2\rangle$ and from state $|1\rangle$ to state $|3\rangle$ when measurement of population of either $|\psi\rangle=|1\rangle$ or $|\psi\rangle=|2\rangle$ is performed are defined by setting $\rho(-T_1)=|1\rangle\langle 1|$ in~(\ref{state}):
\begin{align}\label{Eq:Dynamic}
P_{1\to k}(f_1,{\cal M}_{|j\rangle\langle j|}, f_2)&=\langle k|U_{f_2}{\cal M}_{|j\rangle\langle j|}(U_{f_1}|1\rangle\langle 1| U^\dagger_{f_1}) U^\dagger_{f_2}|k\rangle,
\end{align}
where $k=2,3$ and $j=1,2$.  Numerically constructed kinematic control landscapes are plotted on Fig.~\ref{FigLandscape}. Below we perform the theoretical analysis of these control landscapes.

\begin{figure}[t!]
\includegraphics[width=\linewidth]{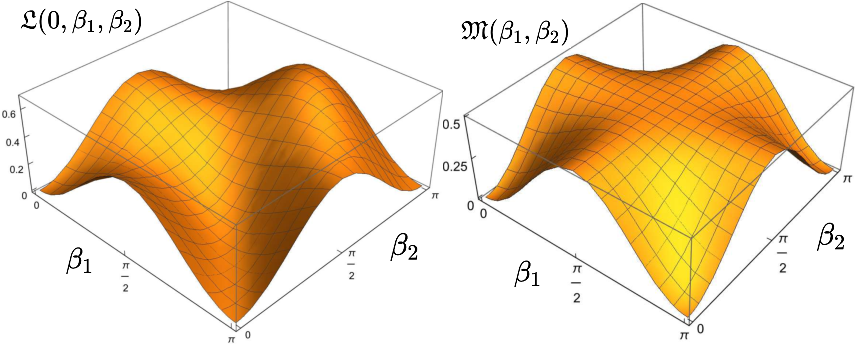}
\caption{Kinematic control landscapes of the transition probability $P_{1\to 2}(U(1),{\cal M}_{|i\rangle\langle i|}, U(2))$. Left: $i=1$ and $\Omega=0$ (i.e., plot of the function $\mathfrak{L}_1(0, \beta_1, \beta_2)$; see below). Right: $i=2$ (i.e., plot of the function $\GG(\beta_1, \beta_2)$; see below). Left landscape has second order-traps (see Theorem~\ref{theorem1}), whereas right landscape is completely trap-free (see Theorem~\ref{theorem3}).}
\label{FigLandscape}
\end{figure}

\section{Spin-1 representation theory}\label{Sec3}
Since we study kinematic control landscape of the system under two coherent controls, we need to consider the manifold
\begin{equation}\label{KinematicManifold}
    \mathcal{R}\times\mathcal{R}, \quad \mathcal{R} = \{SU \in \mathcal{U}(3)\, | \, U \in \text{spin-1 representation of } \mathcal{SU}(2) \text{ Lie group} \}.
\end{equation}
Now, we describe parameterization of $\mathcal{R}$.

For this system, the free and interaction Hamiltonians generate a spin-1 representation of the Lie algebra $\mathfrak{su}(2)$. Let $J_x$, $J_y$ and $J_z$ be standard generators of spin-$1$ representation of the Lie algebra $\mathfrak{su}(2)$ with matrices
\begin{equation*}
J_x = \frac{1}{\sqrt{2}}\begin{pmatrix}
0 & 1 & 0 \\
1 & 0 & 1 \\
0 & 1 & 0
\end{pmatrix}, \quad
J_y = \frac{1}{\sqrt{2}}\begin{pmatrix}
0 & -i & 0 \\
i & 0 & -i \\
0 & i & 0
\end{pmatrix}, \quad J_z = \begin{pmatrix}
1 & 0 & 0 \\
0 & 0 & 0 \\
0 & 0 & -1
\end{pmatrix},
\end{equation*}
which are written in the eigenbasis of $J_z$ denoted by $\{\ket{1, -1}, \ket{1, 0}, \ket{1, 1}\}$, first value denote spin, second value denote spin projection. Below we will use the notation $\ket{1, -1} \equiv \ket{1}, \ket{1, 0}  \equiv \ket{2}, \ket{1, 1} \equiv \ket{3}$.

In the kinematic representation, we consider allowed evolution operators as given by Wigner's D-matrices which are parameterized by Euler angles $\alpha, \beta, \gamma$~\cite{SakuraiBook,WignerBook}.

The first parameterization that we will use is
\begin{equation*}
U(\alpha, \beta, \gamma)=e^{-i\alpha J_z}\cdot e^{-i\beta J_y}\cdot e^{-i\gamma J_z}\, ,
\end{equation*}
which we call $ZYZ$ parameterization. It is defined for $(\alpha, \beta, \gamma) \in (-\pi, \pi]\times[0, \pi)\times(-\pi, \pi]$. For $\beta = 0$ one has
\begin{equation*}
U(\alpha, 0, \gamma) = \begin{pmatrix}
e^{-i(\alpha+\gamma)} & 0 & 0 \\
0 & 1 & 0 \\
0 & 0  & e^{i(\alpha+\gamma)}
\end{pmatrix},
\end{equation*}
so that different pairs of $\alpha$ and $\gamma$ with the same $\alpha+\gamma$ determine one unitary operator~\cite{KuprovBook}. The subset
\begin{equation} \label{singSet}
    \mathcal{B} := \{U\in \mathcal{R}\,|\, U = \text{diag}(e^{-i\phi}, 1, e^{i\phi}),\,\phi \in [0, 2\pi)\}\subset \mathcal{R}
\end{equation}
has to be parameterized in another way.

The second parameterization is
\begin{equation*}
U(\alpha, \beta, \gamma)=e^{-i\alpha J_y}\cdot e^{-i\beta J_z}\cdot e^{-i\gamma J_y}\, ,
\end{equation*}
which we call $YZY$ parameterization. It is defined for $(\alpha, \beta, \gamma) \in (-\pi, \pi]\times[0, \pi)\times(-\pi, \pi]$ and has the same degeneracy for $\beta=0$.

The two matrix exponents appearing in both parametrizations have the form
\begin{align}\label{expJ_y}
    e^{-i\phi J_y} =& \frac{1}{2}\left(\begin{array}{crr}
    1+\cos \phi & -\sqrt{2} \sin \phi & 1-\cos \phi \\
    \sqrt{2}\sin \phi & \cos \phi & -\sqrt{2}\sin \phi \\
    1-\cos \phi & \sqrt{2} \sin \phi & 1+\cos \phi
    \end{array}\right),\\
    e^{-i\phi J_z} =& \left(\begin{array}{lll}
    e^{-i\phi} & 0 & 0 \\
    0 & 1 & 0 \\
    0 & 0 & e^{i\phi}
    \end{array}\right).\label{expJ_z}
\end{align}

\section{Transition probabilities driven by coherent control assisted by one non-selective measurement}\label{Sec4}

Coherent controls $f_1$ and $f_2$ are called {\it dynamic controls}. The dynamic control landscape is defined as graph of the transition probability $P_{1\to 2}(f_1,f_2)$  considered as a functional of coherent controls $f_1$ and $f_2$. For control landscape analysis it is convenient to consider coherent controls as elements of the $L^2$ space so that the transition probability becomes a functional $P_{1\to 2}:L^2\times L^2\to[0,1]$ (not necessarily a surjection). 

Kinematic controls are  parameters determining allowed unitary evolutions which are elements of $\mathcal R$. Below we consider the transition probabilities as functions of the kinematic controls as in~\cite{ShuangPRA2008}. In the kinematic representation, the system during the first period of the evolution evolves with some unitary evolution operator $U(1)=U^{f_1}$, defined by Eq.~(\ref{SchrEq}) with $f=f_1$, then one performs non-selective measurement (we consider measurement of population of either state $|1\rangle$ or $|2\rangle$), and finally after the measurement the system evolves with another evolution operator $U(2)=U^{f_2}$, defined by Eq.~(\ref{SchrEq}) with $f=f_2$.

The transition probabilities in the kinematic representation, i. e. as functions on $\mathcal{R}\times \mathcal{R}$, become
\begin{align}
P_{1\to k} 
(U(1), {\cal M}_{|j\rangle\langle j|}, U(2)) 
=\langle k|U(2){\cal M}_{|j\rangle\langle j|}[U(1)|1\rangle\langle 1| U^\dagger(1)] U^\dagger(2)|k\rangle
\end{align}
Below, we consider $k=2,3$ and $j=1,2$. We are mainly concentrated on the most non-trivial case $k=2$, $j=1$, for which we study the kinematic landscape in details. For comparison in section~\ref{Sec6} we also compute the maximal transition probabilities for the cases $(k,j)=(2,2), (3,1), (3,2)$ and numerically plot some of the corresponding control landscapes.

The initial system state (before application of the coherent control $f_1$) is $\rho(-T_1)=|1\rangle\langle 1|$. The first control transforms the initial state into
\begin{equation*}
\rho\left(0-\right)=U(1)|1\rangle\langle 1|U^{\dagger}(1).
\end{equation*}

Consider the case $j=1$. After performing non-selective measurement of $P_{\ket{1}}=|1\rangle\langle 1|$, the system density matrix will be
\begin{align*}
\rho\left(0+\right)=&P_{\ket{1}}\rho\left(0-\right)P_{\ket{1}}+(\mathbb I-P_{\ket{1}})\rho\left(0-\right)(\mathbb I-P_{\ket{1}})=\left|U_{11}(1)\right|^2|1\rangle\langle 1|\\
&+\left|U_{21}(1)\right|^2|2\rangle\langle 2|
+\left|U_{31}(1)\right|^2|3\rangle\langle 3|+U_{21}(1)\overline{U_{31}(1)}|2\rangle\langle 3|+\overline{U_{21}(1)}U_{31}(1)|3\rangle\langle 2|.
\end{align*}
Finally, 
\begin{align}\label{TP1}
\bra{k}\rho(T_2)\ket{k}=&\bra{k}U(2)\rho\left(0+\right)U^{\dagger}(2)\ket{k}=|U_{11}(1)|^2|U_{k1}(2)|^2+|U_{21}(1)|^2|U_{k2}(2)|^2\nonumber\\
&+|U_{31}(1)|^2|U_{k3}(2)|^2+2\Re[U_{21}(1)\overline{U_{31}(1)}U_{k2}(2)\overline{U_{k3}(2)}].
\end{align}
    
Consider the case $j=2$. After performing non-selective measurement of $P_{\ket{2}}=|2\rangle\langle 2|$, the system density matrix will be
\begin{align*}
\rho\left(0+\right)=&P_{\ket{2}}\rho\left(0-\right)P_{\ket{2}}+(\mathbb I-P_{\ket{2}})\rho\left(0-\right)(\mathbb I-P_{\ket{2}})=\left|U_{11}(1)\right|^2|1\rangle\langle 1|\\
&+\left|U_{21}(1)\right|^2|2\rangle\langle 2|
+\left|U_{31}(1)\right|^2|3\rangle\langle 3|+U_{11}(1)\overline{U_{31}(1)}|1\rangle\langle 3|+\overline{U_{11}(1)}U_{31}(1)|3\rangle\langle 1|.
\end{align*}
Finally, 
\begin{align}\label{TP2}
\bra{k}\rho(T_2)\ket{k}=&\bra{k}U(2)\rho\left(0+\right)U^{\dagger}(2)\ket{k}=\left|U_{11}(1)\right|^2\left|U_{k1}(2)\right|^2+\left|U_{21}(1)\right|^2\left|U_{k2}(2)\right|^2\nonumber\\
&+\left|U_{31}(1)\right|^2\left|U_{k3}(2)\right|^2+2\Re[U_{11}(1)\overline{U_{31}(1)}U_{k2}(2)\overline{U_{k3}(2)}].
\end{align}

\section{Kinematic quantum control landscape of the transition probability}\label{Sec5}
The kinematic control landscape is a function on $\mathcal{R}\times \mathcal{R}$ and from Eq.~(\ref{TP1}) we have that the kinematic landscape of the transition probability is 
\begin{align}
P_{1\to 2}(U(1),{\cal M}_{|1\rangle\langle 1|}, U(2))=&\left|U_{11}(1)\right|^2\left|U_{21}(2)\right|^2+\left|U_{21}(1)\right|^2\left|U_{22}(2)\right|^2+\left|U_{31}(1)\right|^2\left|U_{23}(2)\right|^2\nonumber\\
&+2\Re[U_{21}(1)\overline{U_{31}(1)}U_{22}(2)\overline{U_{23}(2)}].\label{Eq:Kinematic}
\end{align}
Since we consider the 3-dimensional representation of $\mathcal{SU}(2)$ Lie group, we choose Euler angle parameterization (see Sec.~\ref{Sec3}) to study the kinematic control landscape. The parameters $(\alpha, \beta, \gamma)$ in the unitary evolution operator~(\ref{Ukinematic}) are called {\it kinematic controls}. We denote $c_i=(\alpha_i, \beta_i, \gamma_i)$ for $i=1,2$. Thus transition probability is considered as a function of the kinematic control parameters $(\alpha_1, \beta_1, \gamma_1, \alpha_2, \beta_2, \gamma_2)$ on the corresponding open charts. 

First we consider $ZYZ$ parameterization to find critical points on $\left(\mathcal{R}\backslash\mathcal{B}\right)\times \left(\mathcal{R}\backslash\mathcal{B}\right)$. From~(\ref{expJ_y}),~(\ref{expJ_z}) we get the following matrix representation of the evolution operator
\begin{equation}\label{Ukinematic}
U(\alpha, \beta, \gamma) = \begin{pmatrix}
e^{-i(\alpha+\gamma)}\cos^2{\frac{\beta}{2}} & -\frac{e^{-i\alpha}}{\sqrt{2}}\sin{\beta} & e^{-i(\alpha-\gamma)}\sin^2{\frac{\beta}{2}} \\
\frac{e^{-i\gamma}}{\sqrt{2}}\sin{\beta} & \cos{\beta} & -\frac{e^{i\gamma}}{\sqrt{2}}\sin{\beta} \\
e^{i(\alpha-\gamma)}\sin^2{\frac{\beta}{2}} & \frac{e^{i\alpha}}{\sqrt{2}}\sin{\beta}  & e^{i(\alpha+\gamma)}\cos^2{\frac{\beta}{2}}
\end{pmatrix}.
\end{equation}
Using~(\ref{Ukinematic}) we obtain
\begin{align*}
& P_{1\to 2}(c_1,{\cal M}_{|1\rangle\langle 1|}, c_2)=\\
&=
\frac{1}{8}\sin^2{\beta_2}\left(3+\cos{2\beta_1}\right)+\frac{1}{2}\sin^2{\beta_1}\cos^2{\beta_2}-\frac{1}{2}\cos{\Omega}\sin{\beta_1}\sin^2{\frac{\beta_1}{2}}\sin{2\beta_2}=:\LL(\Omega,\beta_1, \beta_2),
\end{align*}
where $\Omega=\alpha_1+\gamma_2$ and $\LL$ is the control landscape function. Hence the transition probability $P_{1\to 2}(c_1,{\cal M}_{|1\rangle\langle 1|}, c_2)$ is a function of only three kinematic parameters $\Omega,\beta_1,\beta_2$. 

The domain of the graph is $(\Omega,\beta_1,\beta_2)\in(-\pi, \pi]\times(0, \pi)\times(0, \pi)$. For $\Omega=0$, the kinematic landscape of the transition probability $P_{1\to 2}(c_1,{\cal M}_{|1\rangle\langle 1|}, c_2)=\LL(0,\beta_1, \beta_2)$ is shown on Fig.~\ref{FigLandscape}. As mentioned in the Introduction, in~\cite{ShuangPRA2008} it was shown that
\be\label{Eq:Max1}
\max\limits_{c_1,c_2} P_{1\to 2}(c_1,{\cal M}_{|1\rangle\bra{1}}, c_2)= 0.004\times(\sqrt{393-48\sqrt{6}}+138+7\sqrt{6}).
\ee

\begin{Lemma}\label{theorem1} All critical points of $\LL(\Omega,\beta_1, \beta_2)$, their types and corresponding values of the function are given in Table~\ref{Table1}.
\begin{center}
\begin{table}
%\captionsetup{justification=centering,margin=1cm}
\begin{tabular}{|p{8.7cm}|c|c|} 
\hline
$(\Omega,\beta_1, \beta_2)$ & $P_{1\to 2}$ & type of the point \\ 
 \hline
 $\left(\pm \frac{\pi}{2}, \frac{\pi}{2}, \frac{\pi}{2}\right)$ & $1/4$ & saddle point \\ 
 \hline
 $\left(\pi, \arccos\left(\frac{-1+\sqrt{5}}{2}\right), \pi-\frac{1}{2}\arctan\left(2\sqrt{2+\sqrt{5}}\right)\right)$,
 $\left(0, \arccos\left(\frac{-1+\sqrt{5}}{2}\right), \frac{1}{2}\arctan\left(2\sqrt{2+\sqrt{5}}\right)\right)$ & $1/4$ & saddle point \\ 
 \hline
  $\left(\pi, \arccos\left(\frac{1+\sqrt{6}}{5}\right), \frac{1}{2}\arccos\left(\frac{1}{1-\sqrt{6}}\right)\right)$, $\left(0, \arccos\left(\frac{1+\sqrt{6}}{5}\right), \pi-\frac{1}{2}\arccos\left(\frac{1}{1-\sqrt{6}}\right)\right)$  & $\frac{3}{50}(9-\sqrt{6}) \approx 0.393$& saddle point \\ 
 \hline
$\left(\pi, \arccos\left(\frac{1-\sqrt{6}}{5}\right), \frac{1}{2}\arccos\left(\frac{1}{1+\sqrt{6}}\right)\right)$, $\left(0, \arccos\left(\frac{1-\sqrt{6}}{5}\right), \pi-\frac{1}{2}\arccos\left(\frac{1}{1+\sqrt{6}}\right)\right)$ & $\frac{3}{50}(9+\sqrt{6}) \approx 0.687$ & global maximum \\ 
\hline
\end{tabular}
\caption{Critical points,  their types and values of the objective function  $\LL(\Omega,\beta_1, \beta_2)$.}
\label{Table1}
\end{table}
\end{center}
\end{Lemma}

\begin{proof} Let us find critical points of the function $\LL(\Omega,\beta_1, \beta_2)$ on the domain $(-\pi, \pi]\times(0, \pi)\times(0, \pi)$. Derivatives of $\LL(\Omega,\beta_1, \beta_2)$ are:
\begin{equation*}
\begin{cases}
    \LL(\Omega,\beta_1, \beta_2)^{\prime}_{\Omega}=\frac{1}{2}\sin{\Omega}\sin{\beta_1}\sin^2{\frac{\beta_1}{2}}\sin{2\beta_2}\\
    \LL(\Omega, \beta_1, \beta_2)^{\prime}_{\beta_1}=\frac{1}{8}\left(1+3\cos{2\beta_2}\right)\sin{2\beta_1}-\frac{1}{4}\cos{\Omega}\left(\cos{\beta_1}-\cos{2\beta_1}\right)\sin{2\beta_2}\\
    \LL(\Omega, \beta_1, \beta_2)^{\prime}_{\beta_2}=\frac{1}{8}\left(1+3\cos{2\beta_1}\right)\sin{2\beta_2}-\frac{1}{2}\cos{\Omega}\sin{\beta_1}\cos{2\beta_2}\left(1-\cos{\beta_1}\right)
    \end{cases}.
\end{equation*}
For the later analysis we will need the Hessian of $\LL(\Omega, \beta_1, \beta_2)$:
\begin{equation*}
    \begin{cases}
    \LL(\Omega,\beta_1, \beta_2)^{\prime\prime}_{\Omega\Omega}=\frac{1}{2} \cos {\Omega} \sin ^2{\frac{\beta_1}{2}} \sin {\beta_1} \sin {2 \beta_2}\\
    \LL(\Omega,\beta_1, \beta_2)^{\prime\prime}_{\Omega\beta_1}=\sin {\Omega} \sin ^2{\frac{\beta_1}{2}} \sin {\beta_2} \left(2 \cos {\beta_1}+1\right) \cos {\beta_2}\\
    \LL(\Omega,\beta_1, \beta_2)^{\prime\prime}_{\Omega\beta_2}=\sin {\Omega} \sin ^2{\frac{\beta_1}{2}} \sin {\beta_1} \cos {2 \beta_2}\\
    \LL(\Omega,\beta_1, \beta_2)^{\prime\prime}_{\beta_1\beta_1}=\frac{1}{4} \left[\cos {\Omega} \left(\sin {\beta_1}-2 \sin {2 \beta_1}\right) \sin {2 \beta_2}+\cos {2 \beta_1}
   \left(3 \cos {2 \beta_2}+1\right)\right]\\
    \LL(\Omega,\beta_1, \beta_2)^{\prime\prime}_{\beta_1\beta_2}=-\cos {\Omega} \sin ^2{\frac{\beta_1}{2}} \left(2 \cos {\beta_1}+1\right) \cos {2 \beta_2}-\frac{3}{4} \sin {2
   \beta_1} \sin {2 \beta_2}\\
    \LL(\Omega,\beta_1, \beta_2)^{\prime\prime}_{\beta_2\beta_2}=2 \cos {\Omega} \sin {\beta_1} \sin {2 \beta_2} \sin ^2{\frac{\beta_1}{2}}+\frac{1}{4} \left(3 \cos {2
   \beta_1}+1\right) \cos {2 \beta_2}
    \end{cases}.
\end{equation*}
We search for all the points where gradient of $\LL$ is zero, i.e. $\nabla \LL(\Omega,\beta_1, \beta_2) = 0$. From the condition $\LL(\Omega,\beta_1, \beta_2)^{\prime}_{\Omega} = 0$ we obtain that such extreme points should lie on the surfaces
\begin{equation*}
    \left[
        \begin{array}{ll}
            \Omega = 0, \pi \\
            \beta_2 = \frac{\pi}{2}
        \end{array} .
    \right .
\end{equation*}

Consider these cases separately.
\begin{enumerate}
    
    \item ${\beta_2=\frac{\pi}{2}}$. In this case 
    \begin{equation*}
        \begin{cases}
            \LL(\Omega, \beta_1, \frac{\pi}{2})^{\prime}_{\beta_1}=-\frac{1}{4}\sin{2\beta_1}\\
            \LL(\Omega, \beta_1, \frac{\pi}{2})^{\prime}_{\beta_2}=\cos{\Omega}\sin{\beta_1}\sin^2\frac{\beta_1}{2}
        \end{cases}.
    \end{equation*}
    Thus from the conditions $\LL(\Omega, \beta_1, \frac{\pi}{2})^{\prime}_{\beta_1}=\LL(\Omega, \beta_1, \frac{\pi}{2})^{\prime}_{\beta_2}=0$ we find the  critical points $(\pm \frac{\pi}{2}, \frac{\pi}{2}, \frac{\pi}{2})$. Hessian has the form
    \begin{align*}
        H_{\LL}\left(-\frac{\pi}{2}, \frac{\pi}{2}, \frac{\pi}{2}\right)=&\left(
    \begin{array}{ccc}
     0 & 0 & \frac{1}{2} \\
     0 & \frac{1}{2} & 0 \\
     \frac{1}{2} & 0 & \frac{1}{2} \\
    \end{array}
    \right),\\
    H_{\LL}\left(\frac{\pi}{2}, \frac{\pi}{2}, \frac{\pi}{2}\right)=&\left(
    \begin{array}{ccc}
     0 & 0 & -\frac{1}{2} \\
     0 & \frac{1}{2} & 0 \\
     -\frac{1}{2} & 0 & \frac{1}{2} \\
    \end{array}
    \right).
    \end{align*}
    
    One has $\LL(\pm \frac{\pi}{2},\frac{\pi}{2},\frac{\pi}{2}) = \frac{1}{4}$. Computing Hessian eigenvalues allows to conclude these are saddle points.
    
    \item $\Omega=\pi$. 
    
    In this case 
    \begin{equation}\label{Eq:Case7}
    \begin{cases}
        \LL(\pi, \beta_1, \beta_2)^{\prime}_{\beta_1}=\frac{1}{8}\left(1+3\cos{2\beta_2}\right)\sin{2\beta_1}+\frac{1}{4}\left(\cos{\beta_1}-\cos{2\beta_1}\right)\sin{2\beta_2}=0\\
        \LL(\pi, \beta_1, \beta_2)^{\prime}_{\beta_2}=\frac{1}{8}\left(1+3\cos{2\beta_1}\right)\sin{2\beta_2}+\frac{1}{2}\sin{\beta_1}\cos{2\beta_2}\left(1-\cos{\beta_1}\right)=0
        \end{cases}.
    \end{equation}
    
    which is equivalent to 
    \begin{equation*}
        \begin{cases}
        \left(1+3\cos{2\beta_2}\right)\sin{2\beta_1}=2\left(\cos{2\beta_1}-\cos{\beta_1}\right)\sin{2\beta_2}\\
        4\sin{\beta_1}\cos{2\beta_2}\left(\cos{\beta_1}-1\right)=\left(1+3\cos{2\beta_1}\right)\sin{2\beta_2}
        \end{cases}.
    \end{equation*}
    Dividing second equation by first, we obtain
    \begin{equation*}
        \frac{\cos{\beta_1}(1+3\cos{2\beta_2})}{2\cos{2\beta_2}(\cos{\beta_1}-1)}=\frac{2(\cos{2\beta_1}-\cos{\beta_1})}{1+3\cos{2\beta_1}},
    \end{equation*}
    from which we can express $1/\cos{2\beta_2}$ in terms of $\cos{\beta_1}\equiv x$ as
    \begin{equation*}
        \frac{1}{\cos{2\beta_2}}=\frac{4(\cos{2\beta_1}-\cos{\beta_1})(\cos{\beta_1}-1)}{\cos{\beta_1}(1+3\cos{2\beta_1})} - 3=\frac{2(2x^2-1-x)(x-1)}{x(3x^2-1)}-3.
    \end{equation*}
    We notice, that $\cos{\beta_1}=1$ (which means that $\beta_1=0$) satisfies~(\ref{Eq:Case7}) for $\beta_2=\frac{\pi}{2}$.
    
    Let us express $\tan{2\beta_2}$ from second equation in~(\ref{Eq:Case7})
\begin{equation*}
\tan{2\beta_2}=4\sin{\beta_1}\frac{\cos{\beta_1}-1}{1+3\cos{2\beta_1}}=2\sqrt{1-x^2}\frac{x-1}{3x^2-1}.
\end{equation*}
    Using the trigonometric identity
    \begin{equation*}
        \frac{1}{\cos^2{2\beta_2}}=1+\tan^2{2\beta_2},
    \end{equation*}
    we obtain the following equation for $x$
    \begin{equation}\label{Eq:star}
        \left(\frac{2(2x^2-1-x)(x-1)}{x(3x^2-1)}-3\right)^2=1+4\frac{(x-1)^2(1-x^2)}{(3x^2-1)^2}.
    \end{equation}
    (We also need to consider $\cos{2\beta_1}=-\frac{1}{3}$ and $\cos{\beta_1}=0$, to make sure that we have not lost any solutions. From the  system~(\ref{Eq:Case7}) we notice that these two cases are not the solutions for any $\beta_2$).
    
    Equation~(\ref{Eq:star}) can be solved analytically. Finally we have
    \begin{equation*}
        (x+1)^2\left(x-\frac{-1-\sqrt{5}}{2}\right)\left(x-\frac{-1+\sqrt{5}}{2}\right)\left(x-\frac{1-\sqrt{6}}{5}\right)\left(x-\frac{1+\sqrt{6}}{5}\right)=0.
    \end{equation*}
    
    Finding $(\beta_1, \beta_2)$ from this equation and combining with obtained previously we get the following critical points (which could be either extrema or saddle points)
    \begin{align}\label{Eq:crits}
    \left[
    \begin{array}{ll}
        \left(\beta_1^{I}, \beta_2^{I}\right)=\left(\arccos\left(\frac{1+\sqrt{6}}{5}\right), \frac{1}{2}\arccos\left(\frac{1}{1-\sqrt{6}}\right)\right)\approx (0.810,1.166) \\
        \left(\beta_1^{II}, \beta_2^{II}\right)=\left(\arccos\left(\frac{1-\sqrt{6}}{5}\right), \frac{1}{2}\arccos\left(\frac{1}{1+\sqrt{6}}\right)\right)\approx (1.865,0.638)\\
        \left(\beta_1^{III}, \beta_2^{III}\right)=\left(\arccos\left(\frac{-1+\sqrt{5}}{2}\right), \pi-\frac{1}{2}\arctan\left(2\sqrt{2+\sqrt{5}}\right)\right)\approx(0.905, 2.475)
    \end{array}.
    \right .
    \end{align}
    
    The corresponding values of the objective function are
    \begin{align*}
    \LL(\Omega,\beta_1, \beta_2) =
    \left[
    \begin{array}{ll}
        \frac{1}{4},\quad \left( \pi, \beta_1^{III}, \beta_2^{III}\right)\\
        \frac{3}{50} (9-\sqrt{6}), \quad\left(\pi, \beta_1^{I}, \beta_2^{I}\right)\\
        \frac{3}{50} (9+\sqrt{6}), \quad \left(\pi, \beta_1^{II}, \beta_2^{II}\right)
    \end{array}.
    \right .
    \end{align*}
    
    We compute the second derivatives to determine the type of each of these critical points
    \begin{equation*}
    \begin{cases}
        \LL(\pi, \beta_1, \beta_2)^{\prime\prime}_{\Omega\Omega}=-\frac{1}{2} \sin ^2{\frac{\beta_1}{2}} \sin{\beta_1} \sin {2\beta_2}\\
        \LL(\pi, \beta_1, \beta_2)^{\prime\prime}_{\beta_1\beta_1}=\frac{1}{4} \left[\left(2 \sin {2 \beta_1}-\sin
   {\beta_1}\right) \sin {2 \beta_2}+\cos {2
   \beta_1} \left(3 \cos {2 \beta_2}+1\right)\right]\\
        \LL(\pi, \beta_1, \beta_2)^{\prime\prime}_{\beta_1\beta_2}=\sin ^2{\frac{\beta_1}{2}} \left(2 \cos
   {\beta_1}+1\right) \cos {2 \beta_2}-\frac{3}{4} \sin
   {2 \beta_1} \sin {2 \beta_2}\\
        \LL(\pi, \beta_2, \beta_2)^{\prime\prime}_{\beta_2\beta_2}=\frac{1}{4} \left(3 \cos {2 \beta_1}+1\right) \cos {2
   \beta_2}-2 \sin ^2{\frac{\beta_1}{2}} \sin
   {\beta_1} \sin {2 \beta_2}\\
        \end{cases}.
    \end{equation*}
    (Only non-zero derivatives are written). Hessian matrix in the corresponding points is
\begin{align*}
H_{\LL}(\pi, \beta_1^{I}, \beta_2^{I})&
=\left(
\begin{array}{ccc}
\frac{1}{250} (13 \sqrt{6}-42) & 0 & 0 \\
0 & \frac{1}{100} (23 \sqrt{6}-32) & \frac{1}{50} (-8-13 \sqrt{6}) \\
0 & \frac{1}{50} (-8-13 \sqrt{6}) & \frac{1}{5} (\sqrt{6}-4) \\
\end{array}
\right),\\
H_{\LL}(\pi, \beta_1^{II}, \beta_2^{II})&=
\left(
        \begin{array}{ccc}
         \frac{1}{250} (-42-13 \sqrt{6}) & 0 & 0 \\
         0 & \frac{1}{100} (-32-23 \sqrt{6}) & \frac{1}{50} (13 \sqrt{6}-8) \\
         0 & \frac{1}{50} (13 \sqrt{6}-8) & \frac{1}{5} (-4-\sqrt{6}) \\
        \end{array}
        \right),\\
H_{\LL}(\pi, \beta_1^{III}, \beta_2^{III})&=
\left(
\begin{array}{ccc}
 \frac{1}{4} (7-3 \sqrt{5}) & 0 & 0 \\
 0 & \frac{1}{2} (\sqrt{5}-3) & \frac{1}{4} (1+\sqrt{5}) \\
 0 & \frac{1}{4} (1+\sqrt{5}) & \frac{1}{4} (\sqrt{5}-1) \\
\end{array}
\right).
\end{align*}
From Sylvester's criterion, we conclude that
    \begin{align*}
    \left[
    \begin{array}{ll}
        \LL(\pi,\beta_1^{III}, \beta_2^{III}) = \frac{1}{4} - \textrm{saddle point}\\
        \LL(\pi,\beta_1^{I}, \beta_2^{I}) = \frac{3}{50} (9-\sqrt{6})\approx 0.393 - \text{saddle point}\\
        \LL(\pi,\beta_1^{II}, \beta_2^{II}) = \frac{3}{50} (9+\sqrt{6})\approx 0.687 -\textrm{global maximum}
    \end{array}
    \right .
    \end{align*}
    
    \item $\Omega=0$.
    
    All the critical points are easily found from~(\ref{Eq:crits}):
    \begin{align}
    \left[
    \begin{array}{ll}
        \left(\beta_1^{I}, \beta_2^{I}\right)=\left(\arccos\left(\frac{1+\sqrt{6}}{5}\right), \pi-\frac{1}{2}\arccos\left(\frac{1}{1-\sqrt{6}}\right)\right)\approx (0.810,1.975) \\
        \left(\beta_1^{II}, \beta_2^{II}\right)=\left(\arccos\left(\frac{1-\sqrt{6}}{5}\right), \pi-\frac{1}{2}\arccos\left(\frac{1}{1+\sqrt{6}}\right)\right)\approx (1.865,2.503)\\
        \left(\beta_1^{III}, \beta_2^{III}\right)=\left(\arccos\left(\frac{-1+\sqrt{5}}{2}\right), \frac{1}{2}\arctan\left(2\sqrt{2+\sqrt{5}}\right)\right)\approx(0.905, 2.475)\\
    \end{array}.
    \right .
    \end{align}
    In the same way as for the case $\Omega=\pi$, we compute second derivative to determine type of each of the critical points:
    \begin{equation*}
    \begin{cases}
        \LL(0, \beta_1, \beta_2)^{\prime\prime}_{\Omega\Omega}=\frac{1}{2} \sin ^2{\frac{\beta_1}{2}} \sin{\beta_1} \sin {2\beta_2}\\
        \LL(0, \beta_1, \beta_2)^{\prime\prime}_{\beta_1\beta_1}=\frac{1}{4} \left[\left(\sin {\beta_1}-2 \sin {2 \beta_1}\right) \sin {2 \beta_2}+\cos {2 \beta_1} \left(3 \cos {2
   \beta_2}+1\right)\right]\\
        \LL(0, \beta_1, \beta_2)^{\prime\prime}_{\beta_1\beta_2}=-\frac{3}{4} \sin {2 \beta_1} \sin {2 \beta_2}-\sin ^2{\frac{\beta_1}{2}} \left(2 \cos {\beta_1}+1\right) \cos {2
   \beta_2}\\
        \LL(0, \beta_2, \beta_2)^{\prime\prime}_{\beta_2\beta_2}=2 \sin {\beta_1} \sin {2 \beta_2} \sin ^2{\frac{\beta_1}{2}}+\frac{1}{4} \left(3 \cos {2 \beta_1}+1\right) \cos {2 \beta_2}
        \end{cases}.
    \end{equation*}
The Hessian matrix in the corresponding points is
\begin{align*}
H_{\LL}(0, \beta_1^{I}, \beta_2^{I})&=
\left(
\begin{array}{ccc}
 \frac{1}{250} (13 \sqrt{6}-42) & 0 & 0 \\
 0 & \frac{1}{100} (23 \sqrt{6}-32) & \frac{1}{50} (8+13 \sqrt{6}) \\
 0 & \frac{1}{50} (8+13 \sqrt{6}) & \frac{1}{5} (\sqrt{6}-4) \\
\end{array}
\right),\\
H_{\LL}(0, \beta_1^{II}, \beta_2^{II})&=
\left(
\begin{array}{ccc}
 \frac{1}{250} (-42-13 \sqrt{6}) & 0 & 0 \\
 0 & \frac{1}{100} (-32-23 \sqrt{6}) & \frac{1}{50} (8-13 \sqrt{6}) \\
 0 & \frac{1}{50} (8-13 \sqrt{6}) & \frac{1}{5} (-4-\sqrt{6}) \\
\end{array}
\right),\\
H_{\LL}(0, \beta_1^{III}, \beta_2^{III})&=
\left(
\begin{array}{ccc}
 \frac{1}{4} (7-3 \sqrt{5}) & 0 & 0 \\
 0 & \frac{1}{2} (\sqrt{5}-3) & \frac{1}{4} (-1-\sqrt{5}) \\
 0 & \frac{1}{4} (-1-\sqrt{5}) & \frac{1}{4} (\sqrt{5}-1) \\
\end{array}
\right).
\end{align*}
From Sylvester's criterion we conclude that
\begin{align*}
    \left[
    \begin{array}{ll}
        \LL(0,\beta_1^{III}, \beta_2^{III}) = \frac{1}{4} - \textrm{saddle point}\\
        \LL(0,\beta_1^{I}, \beta_2^{I}) = \frac{3}{50} (9-\sqrt{6})\approx 0.393 - \textrm{saddle point}\\
        \LL(0,\beta_1^{II}, \beta_2^{II}) = \frac{3}{50} (9+\sqrt{6})\approx 0.687 -\textrm{global maximum}
    \end{array}.
    \right .
\end{align*}
\end{enumerate}
\end{proof}
Next, consider parameterization $YZY$ of the first coherent evolution and $ZYZ$ of the second coherent evolution to find critical points on $\left(\mathcal{B}\backslash\{\mathbb{I}\}\right)\times \left(\mathcal{R}\backslash\mathcal{B}\right)$. Likewise, we obtain matrix representation of unitary evolution from which the transition probability will be
\begin{align*}
&P_{1\to 2}(c_1,{\cal M}_{|1\rangle\langle 1|}, c_2)
=\frac{1}{4}\sin^2{\beta_2}\left[1+\Bigl(\cos{\alpha_1}\cos{\gamma_1}-\cos{\beta_1}\sin{\alpha_1}\sin{\gamma_1}\Bigr)\right]^2\\&+\frac{1}{8}\cos^2{\beta_2}\Bigl[4\cos^2{\alpha_1}\sin^2{\gamma_1}+\sin^2{a_1}\Bigl(3+\cos{2\gamma_1}-2\cos{2\beta_1}\sin^2{\gamma_1}\Bigr)\\&+2\cos{\beta_1}\sin{2\alpha_1}\sin{2\gamma_1}\Bigr]+\Re\Bigl[e^{i(\beta_1-\gamma_2)}\cos{\beta_2}\sin{\beta_2}\Bigl(\cos{\frac{\gamma_1}{2}}\sin{\frac{\alpha_1}{2}}+e^{-i\beta_1}\cos{\frac{\alpha_1}{2}}\sin{\frac{\gamma_1}{2}}\Bigr)^2\\&
\times \Bigl(-e^{-i\beta_1}\cos^2{\frac{\gamma_1}{2}}\sin{\alpha_1}+e^{i\beta_1}\sin{\alpha_1}\sin^2{\frac{\gamma_1}{2}}-\cos{\alpha_1}\sin{\gamma_1}\Bigr)\Bigr]=:\mathfrak{L_2}(\alpha_1, \beta_1, \gamma_1, \beta_2, \gamma_2)\equiv\mathfrak{L_2},
\end{align*}
where $(\alpha_i, \beta_i, \gamma_i) \in (-\pi,\pi]\times (0, \pi)\times (-\pi, \pi],\;i=1,2$. Our set $\left(\mathcal{B}\backslash\{\mathbb{I}\}\right)\times \left(\mathcal{R}\backslash\mathcal{B}\right)$ corresponds to $\alpha_1 =\gamma_1 =0$.

\begin{Lemma}\label{lemma2} All critical points of $\mathfrak{L_2}(\alpha_1, \beta_1, \gamma_1, \beta_2, \gamma_2)$ on the surface $\alpha_1=\gamma_1=0$, their types and corresponding values of the function are given in Table~~\ref{Table2}.
\begin{table}
\center
%\captionsetup{justification=centering,margin=1cm}
\begin{tabular}{|c|c|c|} 
\hline
$(\alpha_1, \beta_1, \gamma_1, \beta_2, \gamma_2)$ & $P_{1\to 2}$ & type of the point \\ 
 \hline
$(0, \beta_1, 0, \frac{\pi}{2}, \gamma_2)$, $\beta_1\in (0, \pi)$, $\gamma_2\in (-\pi, \pi]$ & $1/2$ & second order trap \\ 
  \hline
\end{tabular}
\caption{Critical points,  their types and values of the objective function $\mathfrak{L_2}(\alpha_1, \beta_1, \gamma_1, \beta_2, \gamma_2)$.}
\label{Table2}
\end{table}
\end{Lemma}
\begin{proof}
    \begin{equation*}
        \begin{cases}
            \left(\mathfrak{L}_2\right)'_{\alpha_1}|_{\alpha_1,\gamma_1=0}=0 \\
            \left(\mathfrak{L}_2\right)'_{\beta_1}|_{\alpha_1,\gamma_1=0}=0 \\
            \left(\mathfrak{L}_2\right)'_{\gamma_1}|_{\alpha_1,\gamma_1=0}=0 \\
            \left(\mathfrak{L}_2\right)'_{\beta_2}|_{\alpha_1,\gamma_1=0}=\cos{\beta_2}\sin{\beta_2} \\
            \left(\mathfrak{L}_2\right)'_{\gamma_2}|_{\alpha_1,\gamma_1=0}=0 \\
        \end{cases}.
    \end{equation*}
    Thus, from the condition $\nabla \mathfrak{L}_2 = 0$ we have the following critical points $(0, \beta_1, 0, \frac{\pi}{2}, \gamma_2)$.
    Hessian matrix in the corresponding points is
    \begin{align*}
    H_{\mathfrak{L}_2}&=
    \left(
    \begin{array}{ccccc}
     -\frac{1}{2} & 0 & -\frac{1}{2}\cos{\beta_1} & 0 & 0 \\
     0 & 0 & 0 & 0 & 0 \\
     -\frac{1}{2}\cos{\beta_1} & 0 & -\frac{1}{2} & 0 & 0 \\
     0 & 0 & 0 & -1 & 0 \\
     0 & 0 & 0 & 0 & 0 
    \end{array}
    \right).
    \end{align*}
    Since the eigenvalues are $(0, 0, -1, -\frac{1}{2}(1+\cos{\beta_1}), -\frac{1}{2}(1-\cos{\beta_1}))$, we conclude that $(0, \beta_1, 0, \frac{\pi}{2}, \gamma_2)$ is a second order trap.
\end{proof}

Next, consider parameterization $YZY$ of the second coherent evolution and $ZYZ$ of the first coherent evolution to find critical points on $\left(\mathcal{R}\backslash\mathcal{B}\right)\times \left(\mathcal{B}\backslash\{\mathbb{I}\}\right)$. Likewise, we obtain matrix representation of unitary evolution from which transition probability will be
\begin{align*}
    &P_{1\to 2}(c_1,{\cal M}_{|1\rangle\langle 1|}, c_2)
=\frac{1}{2}\sin^2{\beta_1}\Bigl(\cos{\alpha_2}\cos{\gamma_2}-\cos{\beta_2}\sin{\alpha_2}\sin{\gamma_2}\Bigr)^2\\&+\frac{1}{8}\Bigl(3+\cos{2\beta_1}\Bigr)\Bigl[\sin{\beta_2}\Bigl(\cos{\beta_2}\cos{\gamma_2}\sin{2\alpha_2}+\cos^2{\alpha_2}\sin{\gamma_2}\Bigr)\\&+\frac{1}{4}\sin^2{\alpha_2}\Bigl(3+\cos{2\gamma_2}-2\cos{2\beta_2}\sin^2{\gamma_2}\Bigr)\Bigr]+\sin{\beta_1}\sin^2{\frac{\beta_1}{2}}\Bigl(\cos{\alpha_2}\cos{\gamma_2}-\cos{\beta_2}\sin{\alpha_2}\sin{\gamma_2}\Bigr)\\&\times\Bigl[\sin{\alpha_2}\cos{(\alpha_1-\beta_2)}\sin^2{\frac{\gamma_2}{2}}-\cos{\alpha_1}\cos{\alpha_2}\sin{\gamma_2}-\cos{(\alpha_1+\beta_2)}\cos^2{\frac{\gamma_2}{2}}\sin{\alpha_2}\Bigr]\\&=:\mathfrak{L_3}(\alpha_1, \beta_1, \alpha_2, \beta_2, \gamma_2)\equiv\mathfrak{L_3},
\end{align*}
where $(\alpha_i, \beta_i, \gamma_i) \in (-\pi,\pi]\times (0, \pi)\times (-\pi, \pi],\;i=1,2$. The set $\left(\mathcal{R}\backslash\mathcal{B}\right)\times \left(\mathcal{B}\backslash\{\mathbb{I}\}\right)$ corresponds to $\alpha_2 =\gamma_2 =0$.

\begin{Lemma}\label{lemma_3} There are no critical points of $\mathfrak{L_3}(\alpha_1, \beta_1, \alpha_2, \beta_2, \gamma_2)$ on the surface $\alpha_2=\gamma_2=0$.
\end{Lemma}
\begin{proof}
    \begin{equation*}
        \begin{cases}
            \left(\mathfrak{L}_3\right)'_{\alpha_2}|_{\alpha_2,\gamma_2=0}=-\cos{(\alpha_1+\beta_2)}\sin^2{\frac{\beta_1}{2}}\sin{\beta_1} \\
            \left(\mathfrak{L}_3\right)'_{\beta_2}|_{\alpha_2,\gamma_2=0}=0 \\
            \left(\mathfrak{L}_3\right)'_{\gamma_2}|_{\alpha_2,\gamma_2=0}=-\cos{\alpha_1}\sin^2{\frac{\beta_1}{2}}\sin{\beta_1} \\
            \left(\mathfrak{L}_3\right)'_{\alpha_1}|_{\alpha_2,\gamma_2=0}=0 \\
            \left(\mathfrak{L}_3\right)'_{\beta_1}|_{\alpha_2,\gamma_2=0}=\cos{\beta_1}\sin{\beta_1} \\
        \end{cases}
    \end{equation*}
    Thus, from the condition $\nabla \mathfrak{L}_3 = 0$ we see that there are no critical points on the considered domain.
\end{proof}

Finally, consider parameterization $YZY$ of both first and second coherent evolutions to find critical points on $\left(\mathcal{B}\backslash\{\mathbb{I}\}\right)\times \left(\mathcal{B}\backslash\{\mathbb{I}\}\right)$. We won't provide here expression for $P_{1\to 2}(c_1,{\cal M}_{|1\rangle\langle 1|}, c_2)=\mathfrak{L}_4(\alpha_1, \beta_1, \gamma_1, \alpha_2, \beta_2, \gamma_2)=\mathfrak{L}_4$ in this parameterization, since we only interested in critical points and functional values on the surface $\alpha_1=\gamma_1=\alpha_2=\gamma_2=0$ (corresponds to $\left(\mathcal{B}\backslash\{\mathbb{I}\}\right)\times \left(\mathcal{B}\backslash\{\mathbb{I}\}\right)$).

\begin{Lemma}\label{lemma4}
    All critical points of $\mathfrak{L}_4(\alpha_1, \beta_1, \gamma_1, \alpha_2, \beta_2, \gamma_2)$ on the surface $\alpha_1=\gamma_1=\alpha_2=\gamma_2=0$, their types and corresponding values of the function are given in Table~\ref{Table4}.
    \begin{table}
    \center
%\captionsetup{justification=centering,margin=1cm}
    \begin{tabular}{|c|c|c|} 
    \hline
    $(\alpha_1, \beta_1, \gamma_1, \alpha_2, \beta_2, \gamma_2)$ & $P_{1\to 2}$ & type of the point \\ 
     \hline
    $(0, \beta_1, 0, 0, \beta_2, 0)$, $\beta_1, \beta_2 \in (0, \pi)$ & $0$ & global minima \\ 
      \hline
    \end{tabular}
    \caption{Critical points,  their types and values of the objective function $\mathfrak{L}_4(\alpha_1, \beta_1, \gamma_1, \alpha_2, \beta_2, \gamma_2)$.}
    \label{Table4}
    \end{table}
\end{Lemma}
\begin{proof}
One can easily notice that $\nabla \mathfrak{L}_4|_{\alpha_1=\gamma_1=\alpha_2=\gamma_2=0}=0$, so $(0, \beta_1, 0, 0, \beta_2, 0)$ are critical points for every $\beta_1, \beta_2 \in (0, \pi)$.

Hessian matrix in the corresponding points is
    \begin{align*}
    H_{\mathfrak{L}_4}&=
    \left(
    \begin{array}{cccccc}
     1 & 0 & \cos{\beta_1} & 0 & 0 & 0 \\
     0 & 0 & 0 & 0 & 0 & 0 \\
     \cos{\beta_1} & 0 & 1 & 0 & 0 & 0 \\
     0 & 0 & 0 & 1 & 0 & \cos{\beta_2}\\
     0 & 0 & 0 & 0 & 0 & 0\\
     0 & 0 & 0 & \cos{\beta_2} & 0 & 1
    \end{array}
    \right).
    \end{align*}
Its eigenvalues are $(0, 0, 1+\cos{\beta_1}, 1-\cos{\beta_1}, 1+\cos{\beta_2}, 1-\cos{\beta_2})$. Since $\mathfrak{L}_4(0, \beta_1, 0, 0, \beta_2, 0)=0$, point $(0, \beta_1, 0, 0, \beta_2, 0)$ is global minimum for every $\beta_1, \beta_2 \in (0, \pi)$.
\end{proof}

It remains to study $\{\mathbb{I}\}\times\{\mathbb{I}\}$, $\{\mathbb{I}\}\times \left(\mathcal{R}\backslash\{\mathbb{I}\}\right)$, $\left(\mathcal{R}\backslash\{\mathbb{I}\}\right)\times\{\mathbb{I}\}$. Obviously, $\{\mathbb{I}\}\times\{\mathbb{I}\}$ is global minima with $P_{1\to 2}(\mathbb{I},{\cal M}_{|1\rangle\langle 1|}, \mathbb{I})=0$ and  $\{\mathbb{I}\}\times \left(\mathcal{R}\backslash\{\mathbb{I}\}\right)$ determine the same functional as $\left(\mathcal{R}\backslash\{\mathbb{I}\}\right)\times\{\mathbb{I}\}$, i. e. $P_{1\to 2}(\mathbb{I},{\cal M}_{|1\rangle\langle 1|}, U)=P_{1\to 2}(U,{\cal M}_{|1\rangle\langle 1|}, \mathbb{I})$.

Consider $\{\mathbb{I}\}\times \left(\mathcal{R}\backslash\mathcal{B}\right)$. We use $ZYZ$ representation to parameterize this set. Objective function has the following form
\begin{equation}
    P_{1\to 2}(\mathbb{I},{\cal M}_{|1\rangle\langle 1|}, c) = \frac{1}{2}\sin^2{\beta} =: \mathfrak{L}_5(\beta),
\end{equation}
where $\beta \in (0, \pi)$. From the simplicity of $\mathfrak{L}_5(\beta)$, the following lemma obviously follows.

\begin{Lemma}\label{lemma5}
        All critical points of $\mathfrak{L}_5(\beta)$, their types and corresponding values of the function are given in Table~\ref{Table5}.
    \begin{table}
    \center
    %\captionsetup{justification=centering,margin=1cm}
    \begin{tabular}{|c|c|c|} 
    \hline
    $\beta$ & $P_{1\to 2}$ & type of the point \\ 
     \hline
      $\pi/2$ & $1/2$ & second order trap \\ 
      \hline
    \end{tabular}
    \caption{Critical points,  their types and values of the objective function $\mathfrak{L}_5(\beta)$.}
    \label{Table5}
    \end{table}
\end{Lemma}

Consider $\{\mathbb{I}\}\times \left(\mathcal{B}\backslash\{\mathbb{I}\}\right)$. We use $YZY$ representation to parameterize this set. Objective function has the following form
\begin{align*}
    P_{1\to 2}(\mathbb{I},{\cal M}_{|1\rangle\langle 1|}, c)=&\frac{1}{2}\left(\cos^4{\frac{\gamma}{2}}\sin^2{\alpha}+\sin^2{\alpha}\sin^4{\frac{\gamma}{2}}\sin^2{\alpha}+\cos^2{\alpha}\sin^2{\gamma}\right) \\
    &+\frac{1}{4}\left(\frac{1}{2}\cos{\beta}\sin{2\alpha}\sin{2\gamma}-\cos{2\beta}\sin^2{\gamma}\sin^2{\alpha}\right)=: \mathfrak{L}_6(\alpha,\beta, \gamma),
\end{align*}

\begin{Lemma}\label{lemma6}
        All critical points of $\mathfrak{L}_6(\alpha,\beta, \gamma)$ on the surface $\alpha=\gamma =0$ and corresponding values are given in Table~\ref{Table6}.
    \begin{table}
    \center
    %\captionsetup{justification=centering,margin=1cm}
    \begin{tabular}{|c|c|c|} 
    \hline
        $(\alpha,\beta,\gamma)$ & $P_{1\to 2}$ & type of the point \\ 
         \hline
        $(0,\beta,0)$, $\beta \in (0, \pi)$ & $1/2$ & second order trap \\ 
      \hline
    \end{tabular}
    \caption{Critical points,  their types and values of the objective function $\mathfrak{L}_6(\alpha,\beta, \gamma)$.}
    \label{Table6}
    \end{table}
\end{Lemma}
\begin{proof}
    Since $\nabla \mathfrak{L}_6(0, \beta, 0) = 0$, points $(0, \beta, 0)$ are critical for every $\beta \in (0, \pi)$. Hessian matrix in the corresponding points is
    \begin{align*}
    H_{\mathfrak{L}_6}&=
    \left(
    \begin{array}{ccc}
     1 & 0 & \cos{\beta} \\
     0 & 0 & 0 \\
     \cos{\beta} & 0 & 1\\ 
    \end{array}
    \right).
    \end{align*}
    From Sylvester’s criterion we conclude that $(0, \beta,0)$ is a second order trap.
\end{proof}

\begin{remark}
    Types of the critical points of $P_{1\to 2}(U(1),{\cal M}_{|1\rangle\langle 1|}, U(2))$ are the same as for $\mathfrak{L}_1$, $\mathfrak{L}_2$, $\mathfrak{L}_3$, $\mathfrak{L}_4$, $\mathfrak{L}_5$, $\mathfrak{L}_6$.
\end{remark}

\begin{theorem}
    All the critical points of the kinematic control landscape of the trtansition probability $P_{1\to 2}(U(1),{\cal M}_{|1\rangle\langle 1|}, U(2))$ are global minima, saddles, second order traps, and global maxima.

    \begin{table} 
    \hspace{-23mm}
    \begin{tabular}{|p{13.8cm}|c|c|} 
    \hline
    $(U(1),U(2)):(\alpha_1, \beta_1, \gamma_1, \alpha_2, \beta_2, \gamma_2)$ & $P_{1\to 2}$ & type of the point \\ 
     \hline
     $(YZY, YZY)$: $(0, \beta_1^*, 0, 0, \beta_2^*, 0)$ & $0$ & global minima \\ 
      \hline
    $(\mathbb{I}, \mathbb{I}):$ $(0, 0, 0, 0, 0, 0)$ & $0$ & global minima \\ 
      \hline
    $(ZYZ, ZYZ)$: $(\alpha_1^*, \frac{\pi}{2}, \gamma_1^*, \alpha_2^*, \frac{\pi}{2}, \pm \frac{\pi}{2} - \alpha_1^*)$ & $1/4$ & saddle point \\ 
      \hline
      $(ZYZ, ZYZ)$: $(\alpha_1^*, \arccos\left(\frac{-1+\sqrt{5}}{2}\right), \gamma_1^*, \alpha_2^*, \pi-\frac{1}{2}\arctan\left(2\sqrt{2+\sqrt{5}}\right), \pi - \alpha_1^*)$ & $1/4$ & saddle point \\ 
      \hline
      $(ZYZ, ZYZ)$: $(\alpha_1^*, \arccos\left(\frac{-1+\sqrt{5}}{2}\right), \gamma_1^*, \alpha_2^*, \frac{1}{2}\arctan\left(2\sqrt{2+\sqrt{5}}\right), - \alpha_1^*)$ & $1/4$ & saddle point \\ 
      \hline
      $(ZYZ, ZYZ)$: $(\alpha_1^*, \arccos\left(\frac{1+\sqrt{6}}{5}\right), \gamma_1^*, \alpha_2^*, \frac{1}{2}\arccos\left(\frac{1}{1-\sqrt{6}}\right), \pi- \alpha_1^*)$ & $\frac{3}{50}(9-\sqrt{6})$ & saddle point \\
      \hline
      $(ZYZ, ZYZ)$: $(\alpha_1^*, \arccos\left(\frac{1+\sqrt{6}}{5}\right), \gamma_1^*, \alpha_2^*, \pi-\frac{1}{2}\arccos\left(\frac{1}{1-\sqrt{6}}\right), - \alpha_1^*)$ & $\frac{3}{50}(9-\sqrt{6})$ & saddle point \\ 
      \hline
      $(YZY, ZYZ)$: $(0, \beta_1^*, 0, \alpha_2^*, \frac{\pi}{2}, \gamma_2^*)$ & $1/2$ & second order trap \\ 
        \hline
        $(\mathbb{I}, ZYZ):$ $(0, 0, 0, 0, \frac{\pi}{2}, 0)$ & $1/2$ & second order trap \\ 
      \hline
      $(\mathbb{I}, YZY)$: $(0, 0, 0, 0, \beta_2^*, 0)$ & $1/2$ & second order trap \\ 
      \hline
      $(ZYZ, \mathbb{I})$: $(0, \frac{\pi}{2}, 0, 0, 0, 0)$ & $1/2$ & second order trap \\ 
      \hline
      $(YZY, \mathbb{I})$: $(0, \beta_1^*, 0, 0, 0, 0)$ & $1/2$ & second order trap \\ 
      \hline
      $(ZYZ, ZYZ)$: $(\alpha_1^*, \arccos\left(\frac{1-\sqrt{6}}{5}\right), \gamma_1^*, \alpha_2^*, \frac{1}{2}\arccos\left(\frac{1}{1+\sqrt{6}}\right), \pi- \alpha_1^*)$ & $\frac{3}{50}(9+\sqrt{6})$ & global maxima \\ 
      \hline
      $(ZYZ, ZYZ)$: $(\alpha_1^*, \arccos\left(\frac{1-\sqrt{6}}{5}\right), \gamma_1^*, \alpha_2^*,\pi- \frac{1}{2}\arccos\left(\frac{1}{1+\sqrt{6}}\right), - \alpha_1^*)$ & $\frac{3}{50}(9+\sqrt{6})$ & global maxima \\ 
      \hline
    \end{tabular}
    \caption{Critical points,  their types and values of the objective function. Asterisk indicates that the corresponding variable can take any value in its range.}
    \label{TableTheorem}
    \end{table}
\end{theorem}

\begin{corollary}\label{corollary}
    One has
    \begin{align}\label{Eq:Max2}
    \max\limits_{c_1,c_2}P_{1\to 2}(c_1,{\cal M}_{|1\rangle\langle 1|}, c_2)&=\frac{3}{50} (9+\sqrt{6})\, .
    \end{align}
One can check that the expression in the right hand side (r.h.s.) of~(\ref{Eq:Max2}) coincides with the r.h.s. of~(\ref{Eq:Max1}).
\end{corollary}

\section{Comparison with other cases of measurement-assisted transition probabilities} \label{Sec6}
Here for comparison we study other cases of the  transition probabilities.

\begin{theorem}\label{theorem2} One has
\begin{align}
\max\limits_{c_1,c_2}P_{1\to 3}(c_1,{\cal M}_{|1\rangle\langle 1|}, c_2)&=1\label{eq_max3}\\
\max\limits_{c_1,c_2}P_{1\to 3}(c_1,{\cal M}_{|2\rangle\langle 2|}, c_2)&=1\label{eq_max2}
\end{align}
\end{theorem}

\begin{proof}
Here, there is no need to compute $P_{1\to 3}$ explicitly using Euler parametrization. Notice that
\begin{equation*}
    \bra{3}\rho\left(0-\right)\ket{3}=\bra{3}U(1)|1\rangle\langle 1|U^{\dagger}(1)\ket{3}=\left|\bra{3}U(1)\ket{1}\right|^2=\sin^4\frac{\beta_1}{2}.
\end{equation*}
Since we can set $\alpha_2=\beta_2=\gamma_2=0$ (second evolution is trivial) and after measurement we have $\bra{3}\rho\left(0+\right)\ket{3}=\bra{3}\rho\left(0-\right)\ket{3}$, we obtain
\begin{equation*}
    \max\limits_{c_1,c_2}P_{1\to 3}(c_1,{\cal M}_{|1\rangle\langle 1|}, c_2)=\max \limits_{\alpha_1, \beta_1, \gamma_1 \in [0, 2\pi]} \bra{3}\rho\left(0-\right)\ket{3}=\max \limits_{\beta_1 \in [0, 2\pi]}\sin^4\frac{\beta_1}{2} = 1.
\end{equation*}
Similarly, we obtain
\begin{equation*}
    \max\limits_{c_1,c_2}P_{1\to 3}(c_1,{\cal M}_{|2\rangle\langle 2|}, c_2)=\max \limits_{\alpha_1, \beta_1, \gamma_1 \in [0, 2\pi]} \bra{3}\rho\left(0-\right)\ket{3}=\max \limits_{\beta_1 \in [0, 2\pi]}\sin^4\frac{\beta_1}{2} = 1.
\end{equation*}
\end{proof}

\begin{theorem}\label{theorem3} One has
\begin{equation*}
\max\limits_{c_1,c_2}P_{1\to 2}(c_1,{\cal M}_{|2\rangle\langle 2|}, c_2)=\frac{1}{2}.
\end{equation*}
\end{theorem}
\begin{proof}
From~(\ref{TP2}) we get
\begin{align*}  \bra{2}\rho(T_2)\ket{2}=&\left|U_{11}(1)\right|^2\left|U_{21}(2)\right|^2+\left|U_{21}(1)\right|^2\left|U_{22}(2)\right|^2+\left|U_{31}(1)\right|^2\left|U_{23}(2)\right|^2\\&+2\Re[U_{11}(1)\overline{U_{31}(1)}U_{22}(2)\overline{U_{23}(2)}]\, .
\end{align*}
Computing the transition probability gives
\begin{align*}
    P_{1\to 2}(c_1,{\cal M}_{|2\rangle\langle 2|}, c_2)&=\bra{2}\rho(T_2)\ket{2}\\&=\frac{1}{2}\cos^2{\beta_2}\sin^2{\beta_1}+\frac{1}{8}(3+\cos2\beta_1)\sin^2{\beta_2}+\frac{1}{4}\sin^2{\beta_1}\sin^2{\beta_2}\\&=\frac{1}{2}(\sin^2{\beta_1}+\sin^2{\beta_2}+\sin^2{\beta_1}\sin^2{\beta_2})\equiv \GG(\beta_1, \beta_2)\,.
\end{align*}
In this case, the transition probability is a function of only two  kinematic controls $\beta_1$ and $\beta_2$. Notice that $\GG(\beta_1, \beta_2)\leq 1/2$ and $\GG(0, \pi/2)=1/2$. Thus
\begin{align*}
    \max\limits_{c_1,c_2} P_{1\to 2}(c_1,{\cal M}_{|2\rangle\langle 2|}, c_2)=\frac{1}{2}\,.
\end{align*}
\end{proof}

\section{Discussion}\label{SecConclusions} 
In this work, we consider a three-level quantum system with dynamical symmetry which is controlled by coherent control assisted by back-action of non-selective quantum measurement of population of some state. In the absence of the measurement, only $1/2$ of the population can be transferred from the ground to the intermediate state. Maximum measurement assisted transition from the ground to the intermediate state was computed before. In the kinematic representation of the dynamics, controls are Euler angles parameterizing spin-1 representation of the Lie algebra $\mathfrak{su}(2)$. We study in details the kinematic control landscape of the transition probability from the ground to the intermediate state, describing all its critical points. We find what all critical points of the kinematic control landscape are global maxima, global minima, second order traps and saddles. For comparison, we study the transition probability between the ground and highest excited state, as well as the case when both these transition probabilities are assisted by incoherent control implemented by measurement of the intermediate state. The next important question is what types of the critical points are presented in the dynamic control landscape. Let $F_{\rm dyn}: L^2\to \mathcal {R}$ be the mapping which maps any control $f\in L^2$ into the corresponding evolution operator $U_T^f\in\mathcal R$ with some fixed large enough $T$. Let $P^{\rm kin}_{1\to 2}:\mathcal R\times\mathcal R\to [0,1]$ be the kinematic transition probability (defined by Eq.~(\ref{Eq:Kinematic})). The dynamic landscape represents the transition probability as a function of two coherent controls $f_1$ and $f_2$ (defined by Eq.~(\ref{Eq:Dynamic})) and can be considered as composition of the mappings $F_{\rm dyn}$ and $P^{\rm kin}_{1\to 2}$ as $P_{1\to 2}(f_1,f_2)= P_{\rm kin}\circ (F_{\rm dyn}\times F_{\rm dyn})$. In so called regular points, where the Jacobian of the map $F_{\rm dyn}\times F_{\rm dyn}$ from the control space to $\mathcal R\times\mathcal R$ has full rank, the structure of the dynamical control landscape is the same as the structure of the kinematic control landscape (see Theorem~1 in~\cite{Wu_Pechen_Rabitz_Hsieh_Tsou_2008}). Thus our finding implies that at regular points 
the dynamic landscape has the same points as the kinematic landscape. However, at singular points, where the Jacobian is rank deficient, the situation can be different. Thus our results do not exclude the existence of other critical points in the dynamic landscape. Finding such points and establishing the control landscape structure around them is an important problem which requires a special analysis beyond this work. 

\section*{Acknowledgments}
We thank B.O. Volkov, V.N. Petruhanov and S.A. Kuznetsov for useful comments.
This work is supported by the Russian Science Foundation under grant \textnumero~22-11-00330,
\url{https://rscf.ru/en/project/22-11-00330/} and performed in Steklov Mathematical Institute of Russian Academy of Sciences.

\end{document}